\documentclass[11pt,fleqn]{article}

\usepackage[T1]{fontenc}
\usepackage{amssymb}
\usepackage{amsmath}
\usepackage{graphicx}
\usepackage{subfigure}
\usepackage{xparse}
\usepackage{float}
\usepackage{physics}
\usepackage{fullpage}
\usepackage[hidelinks]{hyperref}
\usepackage{qcircuit}
\usepackage{braket}
\usepackage{color}
\usepackage{mdframed}
\usepackage{mathpazo}
\usepackage{todonotes}

\usepackage{thmtools}
\usepackage{thm-restate}

\usepackage{complexity}

\usepackage[capitalise]{cleveref}

\usepackage{tikz}
\usetikzlibrary{shapes,backgrounds}

\usepackage{todonotes}
\newcommand{\CC}{\mathbb{C}}

\newcommand{\II}{\mathbb{I}}

\newcommand{\Aa}{\mathcal{A}}
\newcommand{\Bb}{\mathcal{B}}
\newcommand{\Cc}{\mathcal{C}}

\newcommand{\Gg}{\mathcal{G}}

\newcommand{\id}{\mathbb{I}}

\newcommand{\diag}{\mathrm{diag}}

\newcommand{\true}{\mathsf{True}}
\newcommand{\false}{\mathsf{False}}

\newcommand{\Val}{\mathsf{Val}}

\newcommand{\defeq}{\mathrel{\overset{\makebox[0pt]{\mbox{\normalfont\tiny\sffamily def}}}{=}}}

\usepackage{complexity}

\newcommand{\ETH}{\mathsf{ETH}}

\renewcommand{\exp}{\mathsf{exp}}

\newcommand{\NIC}{\mathsf{NIC}}

\newcommand{\QCSAT}{\mathsf{QCSAT}}

\newcommand{\junk}{\mathsf{junk}}

\newcommand{\qubit}{\CC^{2}}
\newcommand{\qubits}[1]{(\qubit)^{\otimes #1}}

\newcommand{\yes}{\mathsf{yes}}
\newcommand{\no}{\mathsf{no}}

\newcommand{\be}{\begin{equation}}
\newcommand{\ee}{\end{equation}}
\newcommand{\ba}{\begin{array}}
\newcommand{\ea}{\end{array}}
\newcommand{\bea}{\begin{eqnarray}}
\newcommand{\eea}{\end{eqnarray}}

\newcommand{\ra}{\rangle}
\newcommand{\la}{\langle}

\newcommand{\calH}{{\cal H }}

\renewcommand{\eqref}[1]{\textrm{eq.~}(\ref{#1})}
 \newtheorem{theorem}{Theorem}
\newtheorem{definition}[theorem]{Definition}

\newtheorem{lemma}[theorem]{Lemma}

\newtheorem{fact}[theorem]{Fact}

\newenvironment{proof}{\noindent{\bf Proof:} \hspace*{1mm}}{
    \hspace*{\fill} $\Box$ }
\newenvironment{proof_of}[1]{\noindent {\bf Proof of #1:}
    \hspace*{1mm}}{\hspace*{\fill} $\Box$ }

\newenvironment{xalign}{\subequations\align}{\endalign\endsubequations}

\usepackage[margin=1in]{geometry}
\hypersetup{
	colorlinks,
	linkcolor={red!100!black},
	citecolor={blue!100!black},
}

\title{\textbf{The Parameterized Complexity of Quantum Verification}}

\author{
Srinivasan Arunachalam$^{1}$\thanks{\texttt{srinivasan.arunachalam@ibm.com}} \quad 
Sergey Bravyi$^{1}$\thanks{\texttt{sbravyi@us.ibm.com}} \quad 
Chinmay Nirkhe$^{12}$\thanks{\texttt{nirkhe@cs.berkeley.edu}} \quad 
Bryan O'Gorman$^{1}$\thanks{\texttt{bryan.ogorman@ibm.com}} \\ \ \\
\small $^{1}$IBM Quantum, Thomas J Watson Research Center, Yorktown Heights, New York 10598, USA \\
\small $^{2}$Electrical Engineering and Computer Sciences, University of California, Berkeley 94720, USA
}

\date{}

\begin{document}
\maketitle{}

\begin{abstract}
We initiate the study of parameterized complexity of $\QMA$ problems in terms of the number of non-Clifford gates in the problem description. We show that for the problem of parameterized quantum circuit satisfiability, there exists a classical algorithm solving the problem with a runtime scaling exponentially in the number of non-Clifford gates 
but only polynomially with the system size. This result follows from our main result, that for any Clifford + $t$ $T$-gate quantum circuit satisfiability problem, the search space of optimal witnesses can be reduced to a stabilizer subspace isomorphic to at most $t$ qubits (independent of the system size). Furthermore, we derive new lower bounds on the $T$-count of circuit satisfiability instances and the $T$-count of the $W$-state assuming the classical exponential time hypothesis ($\ETH$). Lastly, we explore the parameterized complexity of the quantum non-identity check problem.
\end{abstract}

\section{Introduction}

The solutions to many important computational problems require resources that seem to scale exponentially in the system size in the worst case.
Parameterized complexity refines this phenomenon by identifying one or more parameters along with algorithms that scale exponentially only in the identified parameters (but polynomially in system size). 
In the worst case, these parameters typically scale with system size, but often interesting instances have an intermediate value of the parameter.
In quantum computing, parameterized complexity has been applied to the classical simulation of quantum systems \cite{garcia-ramirez_2014,PhysRevX.6.021043,bravyi2016improved,ogorman:LIPIcs:2019:10402,Bravyi2019simulationofquantum,best-stab-bound}. 
In particular, parameterizing circuits by the number of non-Clifford gates has yielded many state-of-the-art algorithms for classical simulation.
Here, we initiate the study of the parameterized complexity of \emph{non-deterministic} computation, i.e., \emph{quantum verification}. Specifically, we consider $\QMA$ (Quantum Merlin-Arthur\footnote{The complexity class $\QMA$ is the quantum analog of the classical  non-deterministic complexity classes, $\MA$ and $\NP$~\cite{kitaev2002classical}.}) problems parameterized by the number of non-Clifford gates in their verification circuits and obtain non-trivial upper bounds on the classical complexity of finding an optimal witness and the number of qubits required for its~representation.

\subsection{The parameterized complexity of quantum circuit satisfiability}

The first problem we consider is quantum circuit satisfiability ($\QCSAT$), a canonical $\QMA$-complete problem. In a $\QCSAT$ instance, the input is an $s$-gate quantum circuit $U$ acting on $n + m$ qubits followed by the measurement in the standard basis of any $k > 0$ output qubits.\footnote{While one can map the $k$-qubit measurement to a single measurement, this requires the application of a coherent AND logical gate. Implementing this AND gate requires $\Omega(k)$ non-Clifford gates; in the case of parameterized complexity, this cost may be prohibitive. For this reason, we  define the problem in terms of the measurement of multiple qubits in the standard~basis.} The goal in the $\QCSAT$ problem is to estimate the maximal probability that quantum circuit measurement outputs~$1^k$ when run on input states (i.e.,  witnesses) of the form $\ket{\psi} \otimes \ket{0^m}$ for~$\ket{\psi} \in~\qubits{n}$:
\begin{equation}
    \Val \defeq \max_{\ket \psi \in \qubits{n}} \bra{\psi, 0^{m}} U^\dagger \ketbra*{1^k} U \ket{\psi, 0^{m}} \label{eq:qcsat-def}.
\end{equation}
The problem can also be phrased as a promise decision problem in which the goal is to decide if $\Val > a$ (yes instance) or $\Val < b$ (no instance) for $a > b$. The decision problem is known to be $\QMA$-complete when $a = 2/3$ and $b = 1/3$.

To the best of our knowledge, there are no previously known classical algorithms that exploit the structure of the circuit to solve $\QCSAT$ in less than $\exp(n)$ time; a simple exponential time algorithm for calculating $\Val$ can be achieved by searching over the entire Hilbert space of the witness $\ket{\psi}$. One of the main results of our work is that parameterized $\QCSAT$ instances with $t$ $T$-gates for $t \ll n$ can be solved significantly faster than this naive algorithm.
We show that there is a stabilizer subspace isomorphic to $\qubits{t}$ of the the input Hilbert space which contains all optimal witnesses; here an optimal witness is any input state $\psi$ that maximizes the
probability of observing the output $1^k$.

\begin{restatable}{theorem}{thmmain}
    For every $\QCSAT$ instance $U$ with $t \leq n$ T-gates, there is an $n$-qubit Clifford unitary $W$ and a  $t$-qubit state $\ket{\phi}$ such that~$W \left(\ket{\phi} \otimes \ket{0^{n-t}}\right)$ is an optimal input state. Furthermore, the Clifford unitary $W$ only depends on the description of $U$ and can be computed in (classical, deterministic) time $\poly(n,s)$.
\label{thm:main}
\end{restatable}

This insight can be used to construct a faster algorithm for parameterized $\QCSAT$ instances.

\begin{restatable}{theorem}{thmmainalgorithm}
There exists a classical randomized algorithm that takes as input
a parameterized  instance of $\QCSAT$ problem, a precision parameter $\delta>0$,
 and outputs a real random variable
$\xi$ such that 
\be
0\le \xi\le \Val \quad \mbox{and} \quad
\mathrm{Pr}[\xi \ge (1-\delta)\Val]\ge \frac{99}{100}.
\ee
The algorithm has runtime $\poly(n,m,s,t) + O(\delta^{-1} t 2^{t})$.
\label{thm:main-algorithm-version}
\end{restatable}

Quantum circuit satisfiability ($\QCSAT$) is the analog of quantum circuit simulation ($\QMA$- vs. $\BQP$-completeness) in the same way that classical circuit satisfiability is the analog\footnote{Technically, they are the analogs of randomized circuit satisfiability and randomized circuit simulation ($\MA$- vs. $\BPP$-completeness).} of circuit simulation ($\NP$- vs. $\P$-completeness).
Recall that it is widely believed that any classical algorithm for parameterized quantum circuit simulation should have a runtime scaling exponentially in $t$; the current matching upper-bound scales as $2^{\alpha t}$, 
where $\alpha < 0.3963$
for exact simulators~\cite{best-stab-bound} and $\alpha<0.23$ for approximate
simulators~\cite{bravyi2016improved}.
Our result shows that the $\QCSAT$ verification  problem is not much harder than its simulation counterpart; the resulting exponential scaling is worse, but there is no exponential dependence on the instance size or $n$, the size of the witness. 
This is surprising as \emph{classical} circuit satisfiability is believed to require a runtime scaling exponentially with $n$ (the size of the witness) to solve, while classical circuit simulation is trivially solvable in polynomial-time. 
Therefore, it would be reasonable to expect that a parameterized $\QCSAT$ instance would incur a
slowdown scaling exponentially with the witness size $n$ due to non-determinism of the problem and a slowdown scaling exponentially with $t$ due to its quantumness.
We instead show that the exponential time scaling can be brought to scale with only $t$ when $t \ll n$. 
This is the primary power of Theorem \ref{thm:main}: to efficiently reduce the search space of optimal inputs from $n$ qubits to $t$ qubits. 

Our result may seem surprising in view of the earlier work by Morimae et al.~\cite{morimae2015quantum}
that studied a restricted version of the class $\QMA$ where the verifier
can only perform Clifford gates.  In~\cite{morimae2015quantum}, they found that $\QMA$ with a Clifford verifier coincides with the standard $\QMA$. However, the computational model of~\cite{morimae2015quantum} is
different from ours since it allows {\em adaptive} Clifford gates
that can be classically controlled by the outcomes of intermediate measurements.
In contrast, here we consider unitary (non-adaptive) verification circuits with all 
measurements delayed until the end.

Let us briefly sketch the proof of Theorems \ref{thm:main} and \ref{thm:main-algorithm-version};
complete proofs are provided in Section \ref{subsection:omitted-proofs}.
By definition, solving  $\QCSAT$ is equivalent to estimating the largest eigenvalue of an operator $\rho \defeq \bra{0^{m}} U^\dagger \ketbra*{1^k} U \ket{0^{m}}$ acting on $n$ qubits.  Consider first a simple case when there are no $T$ gates, i.e., $t=0$ and $U$ is a Clifford circuit. At a high level, $\rho$ describes a state (unnormalized)
generated by a sequence of Clifford operations:
(1) initializing each qubit in a basis state or a maximally mixed state,
(2) applying a unitary Clifford gate, and (3) post-selectively measuring a qubit
in the standard basis. Such operations are known to have very limited computational power —
they always produce a (mixed) stabilizer-type state~\cite{gottesman-knill-theorem}.
Accordingly, the largest eigenvalue of $\rho$ can be efficiently computed using the standard stabilizer formalism~\cite{fattal2004entanglement}.
Suppose now that $U$ contains $t$ $T$-gates. It is well-known \cite{magic-state} that a $T$-gate can be implemented
by a gadget that includes only Clifford operations and consumes one copy of a magic state
$|A\ra \propto |0\ra+e^{i\pi/4}|1\ra$.
Replacing each $T$-gate in $U$ by this gadget one gets 
$\rho=\mathrm{Tr}_{\{[t]\}}( (\ketbra*{A^t} \otimes \id_n) \rho')$,
where $\rho'$ is  a bipartite stabilizer  state
of $t+n$ qubits and we trace out the first $t$ qubits. Our key technical tool is the characterization of bipartite stabilizer
states \cite{bravyi2006ghz}. This result implies that 
$\rho'=(C_1^\dag \otimes C_2^\dag ) \rho''(C_1 \otimes C_2)$, where
$C_i$ are unitary Clifford operators and $\rho''$  is a tensor product of 
local single-qubit stabilizer states and at most $t$ two-qubit stabilizer
states shared between the two subsystems.
Using this decomposition we are able to show that 
$\rho = C_2^\dag (\rho_{\mathrm{hard}}\otimes \rho_{\mathrm{easy}})C_2$,
where $\rho_{\mathrm{hard}}$ is some (non-stabilizer) state of $t$
qubits and $\rho_{\mathrm{easy}}$ is a stabilizer state of $n-t$ qubits
whose eigevalues are easy to compute.
Thus, $\QCSAT$ reduces to estimating the largest eigenvalue
of the $t$-qubit state~$\rho_{\mathrm{hard}}$. We remark that this theorem also holds if the T-gate is replaced by an arbitrary angle $Z$-rotation (since we use a post-selective magic state injection gadget, in which case it does not matter if the rotation angle is $\pi/8$ or not).
To prove Theorem~\ref{thm:main-algorithm-version}
we make use of the special structure of the state
$\rho_{\mathrm{hard}}$ and reduce the problem
of computing its largest eigenvalue to
computing the largest Schmidt coefficient of a certain pure bipartite state of at most $t$ qubits.
The latter is computed using the power method
with a random starting state~\cite{kuczynski1992estimating}.

\paragraph{Implications for $\QCMA$ vs. $\QMA$.} The description of a circuit generating a quantum state can constitute a classical witness for that state. It is an open question ($\QCMA$ vs. $\QMA$) in complexity theory if all quantum witnesses have efficient classical descriptions \cite{quant-ph/0210077,4262757,DBLP:conf/mfcs/FeffermanK18}. Our work makes progress on the parameterized version of the question by proving that the witness to any $\QCSAT$ instance with $t$ $T$-gates can be constructed with at most $\exp(t)$ $T$-gates as it only requires $t$ qubits to describe. An interesting open question is if our techniques lend themselves to any sub-exponential in $t$ upper-bound on the $T$-count of optimal witness states.

\subsection{Implied lower bounds from the exponential time hypothesis}
Theorem~\ref{thm:main-algorithm-version} provides an upper bound on the runtime of a classical algorithm for $\QCSAT$ with respect to the number of $T$-gates in the verifier.
In conjunction with other complexity-theoretic assumptions, this also implies a \emph{lower bound} on the $T$-count of the verification circuit. One such assumption is the Exponential-Time Hypothesis ($\ETH$), introduced by Impagliazzo and Paturi~\cite{impagliazzo2001complexity}.
Informally, $\ETH$ is the conjecture that classical $k$-SAT requires exponential classical time. By Theorem~\ref{thm:main-algorithm-version}, a verifier circuit for $\QMA$ with $o(n)$ $T$-gates would imply a $2^{o(n)}$-time algorithm for $k$-SAT,  
because $\NP \subset \QMA$. Thus we get the following lower bound.
\begin{restatable}{corollary}{corethtcount}
$\ETH$ implies that any $\QMA$-complete family of Clifford+T verifier circuits must include circuits with $\Omega(n)$ $T$-gates, where $n$ is the size of the witness.
\label{cor:eth-implification-t-count}
\end{restatable}
Interestingly, we can also use $\ETH$ and Theorem~\ref{thm:main-algorithm-version} to get a lower bound on the $T$-count of a quantum circuit that prepares the $m$-qubit $W$-state 
\begin{align}\ket{W_m} = \frac{1}{\sqrt{m}}\big(
\ket{100\cdots0} + \ket{010\cdots0} + \cdots + \ket{000 \cdots 1}\big).\end{align} 
\begin{restatable}{corollary}{corethwstate}
$\ETH$ implies that any Clifford+T circuit $V$ that, when applied to the all-zero state, outputs $\ket{W_m} \otimes \ket{\junk}$ must include $\Omega(m)$ T gates.
\label{cor:eth-implication-w-state}
\end{restatable}

Proofs of both corollaries are provided in Section \ref{sec:eth}. To our knowledge, this is the first lower bound for the $T$-count of state preparation based on complexity-theoretic assumptions. The proof uses the $\NP$-hardness \cite{barahona1982computational} of Hamiltonians of the form  $H = \sum_{\{i, j\} \in E} Z_i Z_j + \sum_{i \in V} Z_i$.
We show that a verifier circuit can be built from only a W state and Clifford operations.

\subsection{The complexity of the non-identity check problem}
The second $\QMA$-complete problem we consider in this work is the Non-Identity Check ($\NIC$) problem: given a classical description of an $n$-qubit quantum circuit $U$, decide if $U$ is close to the identity operation, i.e., is $\min_\phi \norm{U - e^{i\phi}\cdot  \II}$ $\geq c$ or  $\leq s$ for $c-s\ge 1/\poly(n)$, promised one of them is the case.\footnote{We remark that the operator norm is important here; the problem is in $\BQP$ if the goal is to decide if $\|U-\id\|_2$ is small or large, where $\|\cdot \|_2$ is the normalized-$2$ norm~\cite{DBLP:journals/toc/MontanaroW16}.} $\NIC$ was first considered by Janzing et al.~\cite{nic-def} who showed that this problem is $\QMA$-complete, by reducing it to the $\QCSAT$ problem. Subsequently, Ji and Wu~\cite{nic-ji-wu} showed that the $\NIC$ problem remains $\QMA$-complete for even depth-$2$ circuits with \emph{arbitrary} gates. This motivates a natural question: what is the parameterized complexity of $\NIC$? In particular, how does the complexity of $\NIC$ scale with the number of non-Clifford gates? Below we show that $\NIC$ can be solved in polynomial-time for Clifford circuits.%

\begin{restatable}{theorem}{thmnic}
$\NIC$ for Clifford circuits is contained in $\P$.
\end{restatable}
We sketch the proof of this theorem here while a complete proof is provided in Section \ref{sec:nic}.
By definition, solving $\NIC$ for a Clifford circuit $U$ on $n$ qubits is equivalent to estimating the
maximum eigenvalue of a Hamiltonian 
$H=(\id-U)^\dag (\id -U)$, that is, checking if $\|H\|\geq c^2$ or $\leq s^2$. In the  former case, observe that $\Tr(H^p)\geq c^{2p}$ and in the latter case  $\Tr(H^p)\leq 2^n\cdot s^{2p}$. Furthermore, since $c-s\geq 1/\poly(n)$, we can pick $p=\poly(n)$ to make $c^{2p}\gg 2^n\cdot s^{2p}$. 
Thus it suffices to estimate $\Tr(H^p)$ with $p=\poly(n)$.
To this end, observe that $H^p$ can be written as a weighted sum of Clifford powers $U^i$ with $i\in \{-p,\ldots, p\}$.
Thus  $\Tr(H^p)=\sum_{i=-p}^p\alpha_i\Tr(U^i)$ and 
the coefficients $\alpha_i$ can be efficiently computed
using a simple recursive formula. Furthermore, we can efficiently compute $\Tr(U^i)$ for every $i$
 using the identity $\Tr(U^i)=2^n\bra{\Phi^{\otimes n}}
 U^i \otimes \id \ket{\Phi^{\otimes n}}$,
 where $|\Phi\ra$ is the EPR state. 
 The right-hand side is the inner product between two stabilizer states of $2n$ qubits. Such inner products can be computed exactly in time $O(n^3)$,  
 see~\cite{garcia2012efficient,bravyi2016improved}.
 Putting all this together, we can compute $\Tr(H^p)$ exactly in time $\poly(n)$, which determines if this is a ``yes" or ``no" instance of the $\NIC$ problem.  
  
In addition, we also prove that the $\NIC$ problem for constant-depth circuits using gates from a constant-sized gate set is solvable for vanishing completeness parameters. This is in contrast to \cite{nic-ji-wu} which shows constant-depth circuits using gates from a general gate set is $\QMA$-complete.

\begin{restatable}{theorem}{thmfiniteset}
Let $\Gg$ be any constant-sized gate set of 1 and 2 qudit gates and $U$ a quantum circuit of depth at most $t = O(1)$ acting on $n$ qudits of fixed finite local dimension $d$. Let $a < b$ be any parameters such that $a(n) = o(1)$. Then the $\NIC$ problem $(a,b)$ for this restricted class of circuits is in $\P$.
\label{finite-set-p-theorem}
\end{restatable}
  
Our success at understanding parameterized $\QCSAT$ came from our characterization of optimal witness states of parameterized verifier circuits as small linear combinations of stabilizers. However, the eigenvectors of parameterized $\NIC$ circuits $U$ do not have such a characterization. We need to develop new tools to classically describe the eigenvectors of $U$ in order to achieve an equivalent result for the parameterized $\NIC$ problem. Our inability to do so might suggest that the parameterized problem is harder than we previously suspected. 

We conclude with the following intriguing question regarding the parameterized complexity of $\NIC$ problems.
Since the algorithm for $\NIC$ for Clifford circuits breaks down in the presence of a single non-Clifford gate, the question of parameterized complexity of $\NIC$ is left largely open. It may happen that this problem becomes hard (e.g., $\textsf{NP}$-hard) in the presence of a single non-Clifford gate.

\section{Proofs of Theorems \ref{thm:main} and \ref{thm:main-algorithm-version}}
\label{subsection:omitted-proofs}

\thmmain*

\newcommand{\out}{\mathsf{out}}

\begin{proof}
Suppose $U$ is a Clifford+$T$ circuit with 
$c = s - t$ Clifford gates and 
$t$ $T$-gates. 
We assume that $U$ acts on $n+m$ qubits partitioned into 
a witness register of $n$ qubits and an ancilla register of $m$ qubits.
Let $\Pi_{\out}$ be a projector onto the all-ones state $\ketbra{1}^{\otimes k}$ applied to some designated
output register of $k$ qubits. We assume that $\Pi_{\out}$ acts trivially on the remaining $n+m-k$ qubits. 
Define the {\em maximum acceptance probability} of $U$ as
\begin{equation}
\label{pacc}
\pi(U) = \max_\psi \bra{ \psi \otimes 0^m} U^\dagger \Pi_{\out} U \ket{\psi\otimes 0^m}
=\max_\psi \norm{  \Pi_{\out} U \ket{\psi\otimes 0^m} }^2,
\end{equation}
where the maximum is over all normalized $n$-qubit witness states $\ket \psi$.
Let
\be
\pi(U,\psi) := \|  \Pi_{\out} U|\psi\otimes 0^m\ra \|^2.
\ee
We shall implement each $T$-gate in $U$ by the following well-known
post-selection gadget:
$$
\Qcircuit @C=1em @R=1em {
& \gate{T} & \qw
} \qquad \qquad = \qquad \qquad
\Qcircuit @C=1em @R=1em {
& \ctrl{1} & \qw & \qw \\
\lstick{\bra{A}} & \targ & \meter & \cw & \rstick{\ket{0}}
}
\label{eq:t-gadget}
$$
Here, we measure the output qubit in the standard basis and post-select on a measurement outcome of $0$. The gadget consumes one copy of a single-qubit magic state
\be
|A\ra \defeq \frac1{\sqrt{2}}(|0\ra + e^{i\pi/4} |1\ra)
\ee
and, therefore, we can rewrite $\pi(U, \psi)$ as
\be
\pi(U,\psi) = 2^{t}\| \Pi_{\out}^{(1)} C |\psi \otimes 0^m \otimes A^{\otimes t}\ra\|^2.
\ee
Here $C$ is a Clifford circuit acting on $n+m+t$ qubits with $c+t$ gates
(with $c$ gates originating from~$U$ and $t$ CNOT gates originating
from the gadgets) and
$\Pi_{\out}^{(1)}$ is a product of $\Pi_{\out}$ and single-qubit projectors $|0\ra\la 0|$
applied to the second qubit of each $T$-gate gadget.
The extra factor $2^t$ takes into account that each gadget succeeds with the probability $1/2$. Define $\Pi_{\out}^{(2)} = C^\dag \Pi_{\out} C$. Then
\be
\label{Pout2}
\pi(U,\psi) =  2^{t}\| \Pi_{\out}^{(2)}|\psi \otimes 0^m \otimes A^{\otimes t}\ra\|^2.
\ee
Let us say that a projector $\Pi$ acting on $n$ qubits is
a {\em stabilizer projector} if it can be written as
$\Pi=C^\dag (|0^{n-k}\ra \la 0^{n-k}|\otimes \id_k)C$
for some integer $k\in [0,n]$ and some unitary Clifford operator $C$
on $n$ qubits. We shall use the following facts.
\begin{fact}[\cite{gottesman-knill-theorem}]
\label{fact:meas}
Suppose $\Pi$ is a stabilizer projector on $n$ qubits.
Then $\mathrm{Tr}_{\{t\}}((|0\ra\la 0| \otimes \id_{n-1})\Pi)=\sigma \Pi'$
for some $\sigma \in \{0,1,1/2\}$ and some stabilizer projector $\Pi'$
on $n-1$ qubits. One can compute $\sigma$ and $\Pi'$ in time $\poly(n)$.
\end{fact}
\begin{fact}[\bf Bipartite Stabilizer Projectors~\cite{bravyi2006ghz}]
\label{fact:ghz}
Suppose $\Pi$ is a stabilizer projector acting on a bipartite system $LR$
where $L$ and $R$ are arbitrary qubit registers.
Then there exist unitary Clifford operators $C_L$ and $C_R$ acting  
on $L$ and $R$ respectively 
such that 
\be
\Pi= (C_L \otimes C_R)^\dag \Pi' (C_L \otimes C_R)
\ee
where $\Pi'$ is a tensor product of  the following one-qubit and two-qubit stabilizer projectors:
\begin{itemize}
\item Single-qubit projectors $|0\ra\la 0|$ and $I$.
\item Two-qubit projectors $|\Phi^+\ra\la \Phi^+|$ where $|\Phi^+\ra = (|00\ra+|11\ra)/\sqrt{2}$.
\item Two-qubit projectors  $|00\ra\la 00|+|11\ra\la 11|$.
\end{itemize}
Each two-qubit projector acts on one qubit in $L$ and one qubit in $R$.
The above decomposition can be computed in time $poly(|L|+|R|)$.
\end{fact}
By definition, $\Pi_{\out}^{(2)}$ is a stabilizer projector
on $n+m+t$ qubits. Fact~\ref{fact:meas} implies that 
\be
\label{Pout3'}
\la 0^m|\Pi_{\out}^{(2)} |0^m\ra =  \gamma 2^{-r} \Pi_{\out}^{(3)}
\ee
for some $\gamma\in \{0,1\}$, integer $r\in \{0,\ldots,m\}$,
and some stabilizer  projector $\Pi_{\out}^{(3)}$
on $n+t$ qubits. From Eqs.~(\ref{Pout2},\ref{Pout3'}) one gets
\be
\label{Pout3}
\pi(U,\psi) =  \gamma 2^{t-r}\| \Pi_{\out}^{(3)} |\psi \otimes A^{\otimes t}\ra\|^2.
\ee
If $\gamma=0$ then $\pi(U,\psi)=0$ for any witness $|\psi\ra$.
Accordingly, one can choose the Clifford unitary $W$ in the statement of the
theorem arbitrarily. From now on we assume that $\gamma=1$.
Apply Fact~\ref{fact:ghz} to the stabilizer projector $\Pi_{\out}^{(3)}$ and
the partition $[n+t]=LR$ where $L$ is the $n$-qubit witness register
and $R$ is the $t$-qubit magic state register. We get
\be
\Pi_{\out}^{(3)} = (C_L \otimes C_R)^\dag \Pi_{\out}^{(4)} (C_L\otimes C_R)
\ee
where $\Pi_{\out}^{(4)}$ is a tensor product of one-qubit  and two-qubit stabilizer projectors from Fact~\ref{fact:ghz}.
Substituting this into Eq.~(\ref{Pout3}) with $\gamma=1$ one gets
\be
\pi(U,\psi ) = 2^{t-r} \| \Pi_{\out}^{(4)} C_L|\psi\ra  \otimes C_R |A^{\otimes t}\ra\|^2.
\ee
Equivalently,
\be
\pi(U,C_L^\dag \psi ) = 2^{t-r} \| \Pi_{\out}^{(4)} |\psi\ra  \otimes C_R |A^{\otimes t}\ra\|^2.
\ee
By definition, the register $R$ contains $t$ qubits. 
Thus $\Pi_{\out}^{(4)}$ may 
contain at most $t$ two-qubit projectors $|\Phi^+\ra\la \Phi^+|$ and $|00\ra\la 00|+|11\ra\la 11|$.
Indeed, each two-qubit projector must have one qubit in $R$, see Fact~\ref{fact:ghz}.
Partition $L=L'L''$ such that each two-qubit projector that appears in $\Pi_{\out}^{(4)}$ 
has one qubit in~$L''$ and the other qubit in $R$. Then
\be
\label{L''size}
|L''|\le t \qquad \text{ and } \qquad
\Pi_{\out}^{(4)} = \Gamma_{L'} \otimes \Lambda_{L''R}
\ee
for some stabilizer projectors $\Gamma$ and $\Lambda$.
Here the subscripts indicate the registers acted upon by each projector.
Furthermore, $\Gamma$ is a tensor product of single-qubit projectors
$|0\ra\la 0|$ and $I$.
Define an operator
\be
\label{P5}
\Pi_{\out}^{(5)} \defeq {}_R\la A^{\otimes t} |C_R^\dag \Lambda_{L''R} C_R |A^{\otimes t}\ra_R
\ee
acting on the register $L''$. Note that $\Pi_{\out}^{(5)}$ is Hermitian and positive
semi-definite (although $\Pi_{\out}^{(5)}$ might not be a projector). 
Then 
\be
\label{eq_aux1}
\pi(U,C_L^\dag \psi ) = 2^{t-r}  \la \psi | \Gamma_{L'} \otimes \Pi_{\out}^{(5)}|\psi\ra.
\ee
Here the tensor product separates $L'$ and $L''$.
It follows that $|\psi\ra$ is an optimal witness such that $\pi(U)=\pi(U,\psi)$ if
\be
\label{eq_aux2}
C_L|\psi\ra = |\phi_{L'}\ra \otimes |\phi_{L''}\ra,
\ee
where
$|\phi_{L'}\ra$ is a $(+1)$-eigenvector of the projector $\Gamma_{L'}$
and $|\phi_{L''}\ra$ is an eigenvector of $\Pi_{\out}^{(5)}$ with the maximum eigenvalue.
From \cref{L''size} one infers that $|\phi_{L''}\ra$ is a state of at most $t$ qubits.
Since $\Gamma_{L'}$ is a product of $\ketbra{0}$ and $\II$ terms, divide $L'$ into $L_1'$ and $L_2'$ such that
\begin{align}
    \Gamma_{L'} = \ketbra*{0^{\abs{L_1'}}}_{L_1'} \otimes \II_{L_2'}.
\end{align}
Then, the minimizing $\phi_{L'}$ has the form $\phi_{L'} =  \ket{0^{\abs{L_1'}}}_{L_1'} \otimes \ket{\junk}_{L_2'}$ for any state $\ket{\junk}$. To conclude, one can choose an optimal witness state $|\psi\ra$ such that
\be
|\psi\ra = C_L^\dagger \left(|0^{|L_1'|}\ra \otimes \ket{\junk} \otimes |\phi_{L''}\ra \right)
\ee
for some $(\leq t)$-qubit
state $|\phi_{L''}\ra$ and some Clifford operators $C_L$.
This is equivalent to the statement of the theorem. Additionally observe that all the above steps necessary to obtain $W$ can be implemented efficiently since they only involve manipulations with  Clifford circuits and stabilizer projectors. 
\end{proof}

\thmmainalgorithm*

\begin{proof}
First let us introduce some notations.
Let $\calH(n)$ be the set of all normalized $n$-qubit
pure states. 
Given a hermitian $n$-qubit operator $M$, let 
$\lambda_{max}(M)=\max_{\psi \in \calH(n)} \la \psi|M|\psi\ra$ be the maximum eigenvalue of $M$.
Define two-qubit projectors 
\[
\Gamma^{(1)} = |\Phi^+\ra\la \Phi^+|  \quad \mbox{and} \quad 
\Gamma^{(2)}  =  |00\ra\la 00|+|11\ra\la 11|.
\]
Recall that $|\Phi^+\ra = (|00\ra+|11\ra)/\sqrt{2}$.
Given a pair of disjoint $k$-qubit registers $R$ and $L$, let  
$\Gamma^{(i)}_{RL}$ be a $(2k)$-qubit projector
that applies
$\Gamma^{(i)}$ to the $j$-th qubit of $R$ and the $j$-th qubit of $L$ for each $j=1,\ldots,k$.
Below we follow notations introduced in the proof of Theorem~\ref{thm:main}.
Our starting point is the expression for the quantity $\Val$ derived in Eqs.~(\ref{P5},\ref{eq_aux1},\ref{eq_aux2}) thereof, namely
\be
\label{val_repeated1}
\Val = 2^{t-r} \lambda_{max}(\Pi_{\out}^{(5)}), 
\ee
where $\Pi_{out}^{(5)}$ is a positive semi-definite
operator defined in Eq.~(\ref{P5}), namely
\be
\label{val_repeated2}
\Pi_{\out}^{(5)} = {}_R\la A^{\otimes t} |C_R^\dag \Lambda_{L''R} C_R |A^{\otimes t}\ra_R.
\ee
Recall that $R$ and $L''$ are disjoint qubit registers such that $|R|=t$, $|L''|\le t$,
$C_R$ is some Clifford operator
acting on $R$, and $\Lambda_{L''R}$ is a 
product of one- and two-qubit stabilizer
projectors from Fact~\ref{fact:ghz} such that 
each one-qubit projector acts on $R$ and
each two-qubit projector acts on both $R$ and $L''$.
In other words, $\Lambda_{L''R}$ can be written as
\be
\Lambda_{L''R} =\Gamma^{(1)} _{L_1 R_1}
\Gamma^{(2)} _{L_2 R_2} |0\ra\la 0|_{R_3} 
\II_{R_4}
\ee
for some
partitions $L''=L_1 L_2$ and $R=R_1R_2R_3R_4$
with  $|R_1|=|L_1|$ and $|R_2|=|L_2|$.
Below we use notations
$r_i=|R_i|$ and $\ell_i=|L_i|$.

We claim that  $\Pi_{\out}^{(5)}$ 
commutes with any Pauli operator $Z_j$, $j\in L_2$.
Indeed, it suffices to check that $Z_j$ commutes
with $\Lambda_{L''R}$. The latter acts on the register
$L_2$ by a diagonal operator $\Gamma^{(2)}_{L_2R_2}$
which commutes with $Z$-type Pauli operators
confirming the claim. 
Thus one can choose the maximum eigenvector $\psi$ of  $\Pi_{\out}^{(5)}$
such that $Z_j |\psi\ra = (-1)^{\sigma_j}|\psi\ra$ for all $j\in L_2$
and some unknown $\sigma\in \{0,1\}^{\ell_2}$.
Equivalently, $|\psi\ra = |\psi'\ra_{L_1} \otimes|\sigma\ra_{L_2}\equiv |\psi'_{L_1} \otimes \sigma_{L_2}\ra$ for some $\psi'\in \calH(\ell_1)$.
Using Eq.~(\ref{val_repeated1}) one gets
\be
\label{val3}
\Val = 2^{t-r}
\max_{\psi\in \calH(\ell_1)}\;
\max_{\sigma\in \{0,1\}^{\ell_2}}\;
\la \psi_{L_1} \otimes \sigma_{L_2}|
\Pi_{\out}^{(5)}| \psi_{L_1}\otimes \sigma_{L_2}\ra.
\ee
We shall discard the register $L_2$ using the identity
\[
{}_{L_2} \la \sigma | \Gamma^{(2)}_{L_2 R_2} |\sigma\ra_{L_2} = |\sigma\ra\la\sigma|_{R_2}.
\]
Substituting this identity into  Eq.~(\ref{val3}) gives 
\be
\label{val4}
\Val = 2^{t-r}
\max_{\sigma\in \{0,1\}^{\ell_2}}\;
\lambda_{max}(\Pi_{out}^{(6)}(\sigma))
\ee
where $\Pi_{out}^{(6)}(\sigma)$ is a
positive semi-definite operator
acting on the register $L_1$ defined as
\be
\label{P6}
\Pi_{out}^{(6)}(\sigma)= {}_R\la A^{\otimes t} |C_R^\dag \Gamma^{(1)} _{L_1 R_1}
|\sigma\ra\la \sigma|_{R_2} |0\ra\la 0|_{R_3} 
\II_{R_4}
C_R |A^{\otimes t}\ra_R.
\ee
We shall discard the register $L_1$
using the quantum teleportation identity 
\[
{}_{L_1} \la \psi |  \Gamma^{(1)} _{L_1 R_1} |\psi\ra_{L_1} = 2^{-\ell_1} |\psi^*\ra\la \psi^*|_{R_1}
\]
which holds for any state $\psi \in 
\calH(\ell_1)$. Here $\psi^*$ is the complex conjugate of $\psi$ in the standard basis of $\ell_1$ qubits. Using the teleportation identity 
and the definition of $\Pi_{out}^{(6)}(\sigma)$
one can check that
\be
\label{val5}
\la \psi|\Pi_{out}^{(6)}(\sigma)|
\psi\ra =
 2^{-\ell_1} 
\la A^{\otimes t} |C_R^\dag |\psi^*\ra\la \psi^*|_{R_1} 
|\sigma\ra\la \sigma|_{R_2} |0\ra\la 0|_{R_3} 
\II_{R_4}
C_R |A^{\otimes t}\ra_R
\ee
for any state $\psi \in \calH(\ell_1)$.
At this point both registers $L_1$ and $L_2$
have been discarded.
After some algebra 
one can rewrite Eq.~(\ref{val5}) as 
\be
\label{val6}
\la \psi|\Pi_{out}^{(6)}(\sigma)|
\psi\ra =
 2^{-\ell_1} \| {}_{R_1} \la \psi^* |\varphi(\sigma)\ra_{R_1R_4}\|^2,
\ee
where
$|\varphi(\sigma)\ra$ is a state of $R_1R_4$ defined as
\be
\label{varphi(sigma)}
|\varphi(\sigma)\ra= {}_{R_2R_3} \la \sigma_{R_2},0_{R_3}| C_R |A^{\otimes t}\ra_R.
\ee
Since the set $\calH(\ell_1)$ is closed under
the complex conjugation, Eqs.~(\ref{val4},\ref{val6}) give
\be
\label{val7}
\Val = 2^{t-r-\ell_1}
\max_{\sigma\in \{0,1\}^{\ell_2}}\;
\max_{\psi \in \calH(\ell_1)}\; 
\| {}_{R_1} \la \psi |\varphi(\sigma)\ra_{R_1R_4}\|^2.
\ee
At this point the only remaining registers
are $R_1$ and $R_4$.
Clearly, the optimal state $\psi \in \calH(\ell_1)$
that achieves the maximum in Eq.~(\ref{val7})
coincides with the largest eigenvector
of a reduced density matrix
\[
\rho_{R_1}(\sigma) = \mathrm{Tr}_{R_4} |\varphi(\sigma)\ra\la \varphi(\sigma)|.
\]
Thus Eq.~(\ref{val7}) is equivalent to
\be
\label{val8}
\Val = 2^{t-r-\ell_1}
\max_{\sigma\in \{0,1\}^{\ell_2}}\;
\lambda_{max}(\rho_{R_1}(\sigma)).
\ee
Recall that any $t$-qubit Clifford operator
can be efficiently compiled to a Clifford circuit
with $O(t^2)$ one- and two-qubit gates
~\cite{aaronson2004improved}.
Thus one can compute a $t$-qubit state $|\phi\ra:=C_R |A^{\otimes t}\ra_R$
as a vector of complex amplitudes
in time $\poly(t) 2^t$ using the standard state vector simulator of quantum circuits.
We assume that $\phi$ is stored in a classical RAM memory for the rest of the algorithm such that 
any amplitude of $\phi$ can be accessed in time $\poly(t)$.
By definition, $|\varphi(\sigma)\ra$ is obtained from $|\phi\ra$
by  projecting the registers $R_2R_3$ onto the basis state $|\sigma_{R_2},0_{R_3}\ra$, see Eq.~(\ref{varphi(sigma)}).
Equivalently,  $|\varphi(\sigma)\ra$ is obtained from $|\phi\ra$
by selecting a subset of $2^{t-r_2-r_3}$ amplitudes.
Thus
one can compute 
$|\varphi(\sigma)\ra$ as a vector of complex amplitudes
in time $poly(t)2^{t-r_2-r_3}$ for any given $\sigma$.
By definition, $|\varphi(\sigma)\ra$ is 
a state of $t-r_2-r_3$ qubits.
We shall use the following fact.
\begin{lemma}
\label{lemma:power_method}
Suppose $|\varphi\ra$ is a pure $n$-qubit state
specified as a vector of complex amplitudes
and $\delta>0$ is a precision parameter.
Consider a partition $[n]=AB$, where $A$ and $B$ are disjoint qubit registers.
Let $\rho_A = \mathrm{Tr}_B |\varphi\ra\la \varphi|$ be the reduced density matrix of $A$.
There exists a classical randomized algorithm that
runs in time $O(\delta^{-1} n 2^n)$
and outputs a real random variable $\xi$ such that 
\be
0\le \xi \le \lambda_{max}(\rho_A)
\quad \mbox{and} \quad
\mathrm{Pr}[\xi \ge (1-\delta) \lambda_{max}(\rho_A)]
\ge \frac{99}{100}.
\label{power_method}
\ee
\end{lemma}
\begin{proof}
Assume
wlog that $|A|\le |B|$ (otherwise switch $A$ and $B$).
We have $\rho_A = \eta\eta^\dag$, where $\eta$ is a matrix of size $2^{|A|}\times 2^{|B|}$
with matrix elements $\la x|\eta|y\ra = \la x_A y_B|\varphi\ra$.
Given a list of amplitudes of $|\varphi\ra$, one can can compute the matrix of $\eta$
in time $O(2^n)$ since $|A|+|B|=n$. For any $|A|$-qubit state $|v\ra$ the matrix-vector product
$|v\ra \to \eta^\dag |v\ra$ can be computed in time $O(2^n)$.
Likewise, for any $|B|$-qubit state $|w\ra$ the matrix-vector product
$|w\ra \to \eta |w\ra$ can be computed in time $O(2^n)$.
We conclude that the matrix-vector product $|v\ra \to \rho_A|v\ra$
can be computed in time $O(2^n)$. 

We shall compute an estimator $\xi$
satisfying Eq.~(\ref{power_method})
using 
the power method with a random starting state~\cite{kuczynski1992estimating}.
Namely, let $|\phi\ra \in \calH(|A|)$
be a Haar-random state of $A$.
Given a number of iterations $q\ge 2$, the power method outputs an estimator
\[
\xi_q = \frac{\la \psi_q |\rho_A |\psi_q\ra}{\la \psi_q|\psi_q\ra}, \qquad |\psi_q\ra:=\rho_A^{q-1}|\phi\ra.
\]
Clearly, computing $\xi_q$ requries
$q$ matrix-vector multiplications
for the matrix $\rho_A$.
Thus the runtime 
scales as $O(q2^n)$.
Note that $0\le \xi_q\le \lambda_{max}(\rho_A)$
with certainty. 
Theorem~3.1 of~\cite{kuczynski1992estimating}
guarantees that the relative error
\[
\epsilon_q:=\frac{\lambda_{max}(\rho_A)-\xi_q}{\lambda_{max}(\rho_A)}
\]
obeys
\be
{\mathbb E}(\epsilon_q) \le 
\frac{0.871 \log{{(2^{|A|})}}}{q-1}
\le \frac{n}{q}
\ee
for all $q\ge 2$.
Here the expectation value is taken over the
random starting state $\phi$ and we use the natural logarithm.
Since $\epsilon_q$ is a non-negative random variable,
Markov inequality implies that 
$\epsilon_q \le 100\cdot {\mathbb E}(\epsilon_q)$
with the probability at least $99/100$.
Thus the desired estimator $\xi$ satisfying 
Eq.~(\ref{power_method})
can be chosen
as $\xi=\xi_q$ with $q= \lceil 100 n/\delta \rceil$.
\end{proof}

Applying Lemma~\ref{lemma:power_method} to the state $|\varphi(\sigma)\ra$ of $n=t-r_2-r_3$ qubits
with the registers $A=R_1$ and $B=R_4$
one obtains an estimator $\xi(\sigma)$
satisfying
\be
\label{xi_bound}
0\le \xi(\sigma) \le \lambda_{max}(\rho_{R_1}(\sigma))
\quad \mbox{and} \quad
\mathrm{Pr}[ \xi(\sigma) \ge
(1-\delta) \lambda_{max}(\rho_{R_1}(\sigma))]
\ge 99/100.
\ee
The runtime required to 
compute the estimator $\xi(\sigma)$ for any fixed  $\sigma$ is $O(\delta^{-1} t 2^{t-r_2 -r_3})$.
Thus computing the estimators $\xi(\sigma)$ for all
$\sigma \in \{0,1\}^{\ell_2}$ 
takes time $O(\delta^{-1} t 2^{t})$. Here we noted that $r_2=\ell_2$.
We choose the desired estimator $\xi$
approximating the quantity $\Val$ as
\be
\label{xi_def}
\xi = 2^{t-r-\ell_1} \max_{\sigma\in \{0,1\}^{\ell_2}}\; \xi(\sigma).
\ee
Let $\sigma^*\in  \{0,1\}^{\ell_2}$ be the optimal bit string that achieves the maximum
in Eq.~(\ref{val8}) such that 
\be
\label{val9}
\Val = 2^{t-r-\ell_1}
\lambda_{max}(\rho_{R_1}(\sigma^*)).
\ee
From Eq.~(\ref{xi_bound}) one infers
that $\xi \le \Val$ with certainty and
\[
\mathrm{Pr}[\xi\ge (1-\delta)\Val]
\ge 
\mathrm{Pr}[\xi(\sigma^*) \ge (1-\delta)\lambda_{max}(\rho_{R_1}(\sigma^*)]\ge 99/100.
\]
We conclude by noting
that all manipulations  performed in the proof of Theorem~\ref{thm:main} to compute
the Clifford circuit $C_R$
and the stabilizer projector $\Lambda_{L''R}$
can be implemented 
in time $\poly(n,m,s,t)$ using the standard stabilizer
formalism~\cite{gottesman-knill-theorem}.
Thus the total runtime required to compute 
the desired estimator $\xi$ is
\[
\poly(n,m,s,t) + O(\delta^{-1} t 2^{t}).
\]
\end{proof}

In Appendix \ref{sec:alternate-thm} (Theorem \ref{thm:alternate-alg}) we give an alternative algorithm for solving the same problem as Theorem \ref{thm:main-algorithm-version} but using slightly different techniques. It has some advantages over Theorem \ref{thm:main-algorithm-version} which we elaborate in Appendix \ref{sec:alternate-thm}.

\section{Implied lower bounds from the Exponential-Time Hypothesis}
\label{sec:eth}

The Exponential-Time Hypothesis ($\ETH$), introduced by Impagliazzo and Paturi~\cite{impagliazzo2001complexity}, is the conjecture that, informally, (classical) $k$-SAT requires exponential (classical) time.

\begin{definition}[Exponential-Time Hypothesis~\cite{impagliazzo2001complexity}]
Let 
\begin{equation}
s_k \defeq \inf \left\{ \delta: \text{there exists $2^{\delta n}$-time algorithm for solving $k$-SAT}\right\}.
\end{equation}
Then $s_k > 0$ is a constant for all $k \geq 3$.
\end{definition}
It is a stronger assumption than $\P \neq \NP$, which implies just that $k$-SAT requires superpolynomial-time.
We show that $\ETH$, together with \cref{thm:main-algorithm-version}, implies $T$-count lower bounds, which is \cref{cor:eth-implification-t-count}.

\corethtcount*

\begin{proof}
Consider a family of Clifford + $T$-gate $\QMA$ verifier circuits and let $f(n)$ be the number of~$T$ gates in the circuits for $n$-qubit witnesses.
Suppose that $f(n) = o(n)$.
Then by \cref{thm:main-algorithm-version}, each instance can be solved in $O\left(\poly(n) 2^{o(n)}\right)$ deterministic classical time, which contradicts $\ETH$.
The theorem follows from the fact that $\NP \subseteq \QMA$.
\end{proof}

Second, we show that Theorem \ref{thm:main-algorithm-version} along with $\ETH$ implies a linear lower-bound on the $T$-count complexity of generating a $W$ state. This is Corollary \ref{cor:eth-implication-w-state}.

\corethwstate*

The proof will use the following fact, due to Barahona~\cite{barahona1982computational}:
\begin{theorem}[\cite{barahona1982computational}]
\label{thm:positive-quadratic-np-hard}
Estimating, to within inverse polynomial additive error, the ground state energy of the following class of Hamiltonians is NP-hard:
\begin{equation}\label{eq:positive-quadratic-diagonal-hamiltonian}
H = \sum_{\{i, j\} \in E} Z_i Z_j + \sum_{i \in V} Z_i,
\end{equation}
where $G = (V, E)$ is a planar graph with maximum degree 3.
\end{theorem}

\begin{proof_of}{\cref{cor:eth-implication-w-state}}
Consider a circuit $V$ starting from all-zeroes that outputs $\ket{W_m} \otimes \ket{\junk}$ with $t$ T gates.
We will show that this implies a $2^{O(t)}$-time classical algorithm for $k$-SAT, thus proving the \namecref{cor:eth-implication-w-state} by contradiction.

Let $H'$ be a diagonal Hamiltonian of the form of \cref{eq:positive-quadratic-diagonal-hamiltonian} with $m' = O(n)$ terms, where $n$ is the number of vertices in the graph.
Let $H = H' - m' \leq 0$ be $H'$ shifted by a constant so that it's negative semidefinite.
It will be convenient to write $H$ as a sum of $m = 2m'$ terms:
\begin{equation}
H = \sum_{i=1}^m H_i = \sum_{i =1}^{m'} Z_{S_i} - \sum_{i=m'}^{2m'} 1
\end{equation}
where $S_i$ is a set of one or two indices and $Z_S = \prod_{i \in S} Z_i$.

Let $C$ be the circuit consisting of $m$ gates constructed in the following way.
It acts on an $m$-qubit ``control'' register A and an $n$-qubit ``computational'' register B.
For $1 \leq i \leq m'$, each gate $C_i$ implements the $i$-th term of $H$ (either $Z$ or $ZZ$) on the corresponding qubits of the computational register, controlled on the $i$-th qubit of the control register. Because each term of the Hamiltonian is Pauli, the controlled version is Clifford, and so $C$ is a Clifford operator.
For $m' < i \leq 2m'$, the corresponding gate is simply $Z_i$ (in the control register).

Let $U \defeq (V^{\dagger} \otimes I) \cdot C \cdot (V \otimes I)$; $U$ has $2t$ T gates. 
The probability of measuring all zeros on the control register given the input state $U \ket{0^m, \psi}_{A, B}$ is
\begin{xalign}
&\Tr\left[\ketbra{0^m}_A \otimes I_B V^{\dagger} CV\ketbra{0^m, \psi}_{A, B} V^{\dagger} C^{\dagger} V \right]\\
&=\Tr\left[\ketbra{W^m}_A \otimes I_B C\ketbra{W^m, \psi}_{A, B} C^{\dagger}\right]\\
&=\frac1{m^2} \sum_{i, j, k, \ell} \Tr\left[\ket{e_k} \bra{e_{\ell}}_A I_B C  \ket{e_i, \psi} \bra{e_j, \psi}_{A, B} C^{\dagger}\right]\\
&=\frac1{m^2} \sum_{i, j, k, \ell} \Tr\left[\ket{e_k} \bra{e_{\ell}}_A I_B H_i\ket{e_i, \psi} \bra{e_j, \psi}_{A, B} H_j\right]\\
&=\frac1{m^2} \sum_{i, j} \Tr\left[H_i\ketbra{\psi}_{B} H_j\right]=\frac1{m^2} \sum_{i, j} \braket{\psi| H_i H_j| \psi} =\frac1{m^2} \braket{\psi| H^2| \psi}.
\end{xalign}
Because $H$ is negative semidefinite, the state $\psi$ that maximizes the probabality of measuring all zeros as above also minimizes $\braket{\psi | H | \psi}$. By \cref{thm:main-algorithm-version}, such a $W$-state preparation circuit $V$ with $t$ $T$-gates implies a $2^{O(t)}$-time classical algorithm for solving $k$-SAT.
\end{proof_of}

\section{The complexity of quantum non-identity check}
\label{sec:nic}

\begin{definition}[Non-Identity Check \cite{nic-def}]
An instance of the non-identity check ($\NIC$) problem is a classical description of a quantum circuit $U$ and two real numbers $a, b$ such that $b > a$ with the promise~that 
\begin{align}
    d_{\II}(U) \defeq \min_\phi \norm{U - e^{i\phi}\cdot  \II}
\end{align}
is either at most $a$ or at least $b$. 
The instance is called a $\yes$ instance if
$d_{\II}(U)\ge b$ and a $\no$ instance 
if~$d_{\II}(U)\le a$.
\end{definition}

In this section, we present two scenarios in which the non-identity check problem becomes trivial to solve. It was proved by Ji and Wu \cite{nic-ji-wu}, that the problem is in general $\QMA$-hard for depth 2 circuits over qudits when $b - a = 1/\poly(n)$ and the quantum gates are specified to $\Omega(\log n)$ bits of precision. We first show that $\NIC$ is solvable in $\P$ when the entire circuit is Clifford regardless of the depth of the circuit. Second, we show that if $a = o(1)$ is a sub-constant function, then the problem is in $\P$ for constant-depth circuits built from a finite gate set.

\subsection{Clifford circuits}

\thmnic*

\begin{proof}
In order to see this, let $C$ be the unknown Clifford circuit for which we need to determine whether
$\|C-\II\|\geq \alpha$ or $\|C-\II\|\leq \beta$ for $\alpha-\beta \geq 1/\poly(n)$. To this end, consider a Hamiltonian $H=(\II-C)^\dagger (\II-C)=2\II-C-C^\dagger$, whose norm satisfies $\|H\|=\|\II-C\|^2\leq 4$. In order to solve $\NIC(C)$, we need to decide if $\|H\|\geq \alpha^2$ or $\|H\|\leq \beta^2$. In the former case observe that $\Tr(H^p)\geq \alpha^{2p}$ and in the latter case we have $\Tr(H^p)\leq 2^n\cdot \beta^{2p}$. Since $\alpha-\beta \geq 1/\poly(n)$, it suffices to pick $p=\poly(n)$, in order to satisfy $\alpha^{2p}\geq 2\cdot 2^n\cdot \beta^{2p}$. Hence if an algorithm could estimate $\Tr(H^p)$ well enough, then it can distinguish if $C$ satisfies the Yes or No instance of $\NIC$ problem. 

We now show how to compute $\Tr(H^p)$ for $p=\poly(n)$  exactly and efficiently. In this direction, observe that we can express $H^p$ in terms of sum of Clifford powers, i.e.,
\begin{align}
\label{eq:traHp}
\Tr(H^p)=\Tr\big(\sum_{i=-p}^pa_{i}C^{i}\big)=\sum_{i=-p}^pa_{i}\Tr(C^{i})
\end{align}
for some coefficients $a_i\in \R$. Furthermore, since $H$ assumes a simple form $H=\II-C$, the $(2p+1)$ coefficients $a_i$ can be computed explicitly in $\poly(n)$ time. Now, it remains to compute $\Tr(C^i)$ for each one of the $2p+1$ terms. The trace of a Clifford power can be computed exactly by observing that $\Tr(C^i)=2^n\langle \Phi | \II \otimes C^i \vert \Phi \rangle$ where $\ket{\Phi}=\frac{1}{\sqrt{2^n}}\sum_i \ket{i,i}$ is the maximally entangled state on $(2n)$-qubits. Observe that since $\ket{\Phi}$ is a stabilizer state, we have that $\ket{\psi}=\II \otimes C^i  \ket{\Phi}$  is also a $(2n)$-qubit stabilizer state. It remains to estimate inner product between two stabilizer states $\ket{\Phi}$ and $\ket{\psi}$.
It is known that the inner product between any $n$-qubit stabilizer states
can be computed exactly (including the overall phase) in time $O(n^3)$, see~\cite{garcia2012efficient,bravyi2016improved}. Therefore, one can compute each one of the $\Tr(C^i)$ and $a_i$ in \cref{eq:traHp} in time $\poly(n)$ and overall since $p=\poly(n)$, we can exactly compute $\Tr(H^p)$ in time $\poly(n)$. This suffices to decide if $\Tr(H^p)\geq \alpha^{2p}$ (the YES instance of $\NIC$) or $\Tr(H^p)\leq 2^n\cdot \beta^{2p}$ (the NO instance of $\NIC$) for $\alpha-\beta\geq 1/\poly(n)$.
\end{proof}

\subsection{Constant-sized gate sets}

In this section, we are going to show that the $\NIC$ problem for $a = o(1)$, is in $\P$ if we restrict ourselves to circuits of constant depth and a constant-sized gate set. Curiously, the problem was shown to be $\QMA$-hard when either we use an arbitrary gate set or constant-sized gate sets, but we allow circuits of $\Omega(\log n)$-depth \cite{nic-ji-wu}.

\thmfiniteset*

For this proof, we will need a few definitions and preliminary lemmas which we list here first and prove after the proof of the theorem. First, we will need a wonderful fact about low-depth circuits that the reduced density matrix $\tr_{-i} U\psi U^\dagger$ only depends on the lightcone of the $i$th qudit. 

\begin{fact}
Consider a quantum state $\psi$ on $n$ qudits and $U$ a quantum circuit. For any subset $A$ of the qudits, let $L_A$ be the support of the lightcone of $A$ with respect to $U$. Then,
\begin{align}
    \tr_{-A}(U \psi U^\dagger) = \tr_{-A} \left(  U_{L_A} \left( \psi_{L_A} \otimes \nu_{-L_A} \right) U_{L_A}^\dagger \right)
\end{align}
where $\nu$ is the maximally mixed quantum state and $U_{L_A}$ the circuit restricted to gates contained in $L_A$.
\end{fact}

Second, we notice that if a quantum circuit $U$ is close to identity overall, then the reduced action of the circuit on any region must also be close to identity. We will often use the contrapositive of this statement: if the reduced action of a circuit on any small region is far from identity, then the circuit overall is far from identity.

\begin{fact}
Let $U$ be a quantum circuit on $n$ qudits and let $a$ be a constant such that $d_\II(U) < a$. Then for all states $\psi$ and all regions $A$,
\begin{align}
    \norm{\tr_{-A}(U \psi U^\dagger) - \psi_A} \leq a.
\end{align}
\end{fact}

\begin{proof_of}{\cref{finite-set-p-theorem}}
Let $\Cc(\ell, h)$ be the collection of all quantum circuits acting on $\ell$ qudits and of depth $\leq h$ consisting of gates only from $\Gg$. Let us define the \emph{increment-distance} $\eta_{\ell,h}$ as
\begin{align}
    \eta_{\ell,h} \defeq \min_{\substack{V \in \Cc(\ell, h) \\ V \neq e^{i\phi}\II}} d_{\II}(V).
\end{align}
Since $\Gg$ is a finite gate set and $\Cc(\ell, h)$ has a bounded cardinality, $\eta_{\ell,h} > 0$ is a well-defined constant independent on $n$ and represents the closest a circuit can be to being identity without being identity itself. This is formalized in the following fact.

\begin{fact}
Let $V$ be a circuit $\in \Cc(\ell, h)$ such that $d_{\II}(V) < \eta_{\ell, h}$. There exists an angle $\phi_V$ such that $V = e^{i\phi_V} \II$. 
\end{fact}

In order to construct a $\P$ algorithm for this problem, we notice that since $a = o(1)$, for some sufficiently large $N_0$, if $n > N_0$
\begin{align}
    a(n) < \eta_{2^{t+1}, t}.
\end{align}
Our algorithm will solve only instance of this size or larger. Assume that an instance $U$ of the problem is a $\false$ instance, so  $U$ is near-identity. Then for each qubit $i$ and every state~$\psi$,
\begin{align}
    \norm{\tr_{-i}\left( U_{L_i} \left( \psi_{L_i} \otimes \nu_{-L_i} \right) U_{L_i}^\dagger \right) - \psi_i} \leq a
\end{align}
as a consequence of the prior stated facts. Since, this holds for all states $\psi$, then $d_{\II}(U_{L_i}) < a$. However, the circuit $U_{L_A}$ acts on at most $2^{t+1}$ qubits and has depth at most $t$. Since $a < \eta_{2^{t+1}, t}$, then we can conclude that the action of $U_{L_i}$ on the $i$th qubit must be $\II$ (up to phase) in every $\false$ instance. 
Since this holds for every qubit $i$, in a $\false$ instance, the circuit $U$ must exactly be $\II$ (up to phase). Therefore, the $\P$ algorithm is simple: test if each circuit $U_{L_i}$ is exactly identity and if so report $\false$ or otherwise report $\true$.
\end{proof_of}

\section*{Acknowledgments}
SB was supported in part by the IBM Research Frontiers Institute. CN was supported by NSF Quantum Leap Challenges Institute Grant number OMA2016245 and an IBM Quantum PhD internship. Part of this work was completed while CN and BO were participants in the Simons Institute for the Theory of Computing \emph{Summer Cluster on Quantum Compuatation}. 
Additionally, we thank Sam Gunn, Zeph Landau, Dimitri Maslov and Kristan Temme for insightful discussions.  

\bibliography{references}
\bibliographystyle{alpha}

\appendix

\section{Alternative algorithm to Theorem \ref{thm:main-algorithm-version}}
\label{sec:alternate-thm}

In Theorem \ref{thm:alternate-alg}, we give an alternative algorithm for solving the same problem as Theorem \ref{thm:main-algorithm-version}. While this algorithm will have a inferior worst-case runtime than Theorem \ref{thm:main-algorithm-version}, it may run significantly faster depending on the structure of the problem instance.
Furthermore, it has the added advantage that its runtime can be efficiently calculated in time $\poly(n,s)$. Therefore, one can quickly compute the runtime of \cref{thm:alternate-alg} and compare it to that of \cref{thm:main-algorithm-version} and run the faster algorithm. While not faster in a worst-case sense, it may prove optimal for many physical instances. In addition, Theorem \ref{thm:alternate-alg} can handle not just $T$ gates but all single-qubit phase gates $G = \diag(1, e^{i \theta})$ which may be an advantage for some problems\footnote{It is also easy to extend this algorithm to all $k$-qubit non-Clifford gates. However, the runtime will now scale as $4^{t + kt}$. This follows directly from the fact that any $k$-qubit non-Clifford gate can be expressed as the linear combination of $4^k$ Clifford gates.}. 
\footnote{\cref{thm:main-algorithm-version} can also be extended to general phase gates, but at the cost of potentially weaker upper bounds on the constant $\alpha$ in the exponent.}

\begin{theorem}
For $t \leq n$, there exists a classical algorithm for parametrized $\QCSAT$ instance with a single-qubit output register consisting of Clifford and phase gates $G = \diag(1, e^{i\theta})$, running in worst-case time $O(\poly(n,s)2^{3t} \log(t/\delta))$, which decides if $\Val > c$ or $<c - \delta$. Furthermore, there exists a classical $\poly(n,s)$ routine to calculate the runtime of this algorithm without running the algorithm itself.
\label{thm:alternate-alg}
\end{theorem}

\begin{proof}

Let us notice that our goal is to compute the largest eigenvalue of $\bra{0^m} U^\dagger \ketbra{1}_1 U \ket{0^m}$ due to \cref{eq:qcsat-def}. Since $\ketbra{1} = \frac{\II}{2} - \frac{Z}{2}$, this is equivalent to computing the smallest eigenvalue of $\bra{0^m} U^\dagger Z_1 U \ket{0^m}$. In the case that $U$ is a completely Clifford circuit, $Q^{(s)} = U^\dagger Z_1 U$ is a Pauli matrix $Q^{(s)} = \alpha P_1 \otimes P_2 \otimes \ldots P_{n+m}$ for $\alpha \in \{\pm 1, \pm i\}$ and $P_i \in \{\II, X, Y, Z\}$. Therefore,
\begin{align}
    \bra{0^m} U^\dagger Z U \ket{0^m} = P_1 \otimes \ldots \otimes P_n \cdot \prod_{j = n+1}^{n+m} \ev{P_j}{0}.
\end{align}
Then, the smallest eigenvalue of this matrix is easy to calculate as it is in tensor product; this is effectively the Gottesman-Knill theorem \cite{gottesman-knill-theorem}.
Furthermore, the Pauli $P^{(s)}$ can be computed efficiently. More specifically, if $U = g_1 \ldots g_s$ with each gate $g_i$ a Clifford matrix, we can define and compute the sequence of Paulis $Q^{(k)} \defeq g_k Q^{(k-1)} g_k^\dagger$ for $Q^{(0)} = Z_1$ from $k = 1, \ldots, s$.

We now extend this algorithm to the case that $U$ contains $t$ non-Clifford gates. Consider first the case that there is one non-Clifford qubit rotation gate $g_k$, with $g_k$ 
\begin{align}
    g_k \defeq R(\theta_k) = \begin{pmatrix} 1 & 0 \\ 0 & e^{i \theta_k} \end{pmatrix}
\end{align}
and acts (without loss of generality) on the first qubit. Let $Q^{(k-1)}$ be the Pauli calculated up to gate $g_{k-1}$. Now notice that there are 2 cases to consider; namely if the action of $Q^{(k-1)}$ on the first qubit is $\in \{\II, Z\}$ or is $\in \{X, Y\}$. Since $R(\theta)$ commutes with $\II$ and $Z$ and the following commutation relations hold:
\begin{align}
    \quad R(\theta) X R(\theta)^\dagger = (\cos \theta) X + (\sin \theta) Y, \quad R(\theta) Y R(\theta)^\dagger = (-\sin \theta) X + (\cos \theta) Y,
    \label{eq:commutation-relations}
\end{align}
we can express 
\begin{align}
    Q^{(k)} = g_k Q^{(k-1)} g_k^\dagger = P^{(k,1)} + P^{(k,2)}
\end{align} 
where $P^{(k,1)}$ and $P^{(k,2)}$ are Pauli matrices scaled by a complex number. For every subsequent Clifford gate $g_{k'}$ we can then recursively define and compute
\begin{align}
    P^{(k',\ell)} \defeq g_{k'} P^{(k'-1,\ell)} g_{k'}^\dagger
\end{align}
which will remain a Pauli matrix. It is easy to note that this bifurcation from one Pauli matrix to two Pauli matrices when commuting past a non-Clifford gate generalizes to multiple non-Clifford gates with
\begin{align}
    U Z_1 U^\dagger = Q^{(s)} = \sum_{\ell = 1}^{\leq 2^t} P^{(s,\ell)}
    \label{eq:target-sum-in-linear-combo}
\end{align}
being expressible as the linear combination of $\leq 2^t$ Pauli matrices each acting on $n + m$ qubits. We next show that although there are $\leq 2^t$ Pauli matric, there exists an efficiently computable basis of size $b \leq t + 1$ such that each Pauli matrix can be expressed as a product of the basis matrices.

\begin{lemma}
Let $U$ be a quantum circuit consisting of $s$ gates of which at most $t$ gates are non-Clifford qubit rotation gates. Then $Q^{(s)} = U Z_1 U^\dagger$ can be expressed as a sum of $\leq 2^b$ Pauli matrices which are each products of at most $b \leq t + 1$ basis Pauli matrices. Furthermore, the basis can be computed in time $O(\poly(s))$ and the collection of Pauli matrices can be computed in time $O(2^b \cdot \poly(s))$.
\label{lemma:short-basis}
\end{lemma}
This lemma is proved after the description of the rest of the algorithm. Given the basis $\Bb = \{B^{(1)}, \ldots, B^{(b)}\}$ for $b \leq t + 1$, define $\gamma(k', k) \defeq 1$ if $B^{(k')}$ and $B^{(k)}$ commute and $ \defeq 0$ otherwise. Observe then the Pauli matrices
\begin{align}
    A^{(k)} \defeq \prod_{k' < k} X_{k'}^{\gamma(k',k)} \cdot Z_k
\end{align}
observe the same commutation relations as $\Bb$ do. However, $\Aa = \{A^{(1)},\ldots, A^{(b)}\}$ act on at most $b$ qubits. For each Pauli matrix $P^{(\ell)}$ defined as a product of elements from $\Bb$, let us define $O^{(\ell)}$ as the same product, except using the corresponding matrices from $\Aa$. Then the spectrum of 
\begin{align}
    H' \defeq \sum_{\ell = 1}^{\leq 2^t} O^{(\ell)}
\end{align}
is the same as that of $Q^{(s)}$ from \cref{eq:target-sum-in-linear-combo}. This is because there exists a unitary mapping $\Aa$ to $\Bb$ which therefore maps $O^{(\ell)}$ to $P^{(\ell)}$ and likewise maps $H$ to $Q^{(s)}$. Therefore, it suffices to compute the minimum eigenvalue of $H'$. Here $H'$ is a square matrix of dimension $2^b \times 2^b$ with $b \leq t + 1$. Computing $H'$ and its minimum eigenvalue to accuracy $\delta$ can be done in time $O(\poly(s) 2^{3b} \log(t/\delta))$. \\

Lastly, notice that a convenient quality of this algorithm is that the basis $\Bb$ and its size $b$ can be computed in time $O(\poly(s))$. Therefore, one can calculate $b$ and calculate\footnote{Namely, this is to check in time $O(\poly(s))$ if $b \ll t$ to see if this algorithm will be more efficient at computing $H'$ than the stabilizer-rank of magic states algorithm (Theorem \ref{thm:main-algorithm-version}) derived from \cite{best-stab-bound}.} the runtime of the algorithm without running the algorithm itself.
\end{proof} \\

\begin{proof_of}{\cref{lemma:short-basis}}
We proceed by induction on the gates $g_1, \ldots, g_s$ of $U$. Initially, the only basis matrix is $B^{(0,1)} = Z_1$ and $P^{(0,1)} = B^{(0,1)}$. Then  in the inductive step, we assume a basis of $\{B^{(k,\lambda)}\}$ and a collection of Pauli matrices $P^{(k, \ell)}$ such that each Pauli is expressible as a product of the basis matrices. When gate $g_k$ is a Clifford, we define $B^{(k-1,\lambda)} \defeq g_k B^{(k-1,\lambda)} g_k^\dagger$. Conveniently, $P^{(k,\ell)} = g_k P^{(k-1, \ell)} g_k^\dagger$ is the product of the corresponding set of new basis matrices. 

In the case that $g_k = R(\theta)$ which (without loss of generality) acts on the first qubit, we first transform the basis $\{B^{(k-1,\lambda)}\}$ by multiplying basis terms we ensure that at most 2 basis terms act non-trivially on the first qubit. The terms $P^{(k-1,\ell)}$ can be adjusted in polynomial-time to reflect the new basis. If there are no basis terms acting non-trivially or one basis term acting as $Z$, then we set $P^{(k,\ell)}$ equal to $P^{(k-1, \ell)}$. In the case that the one basis term (without loss of generality, $B^{(k-1,1)}$ acts as $X$, then we introduce a new basis term defined as $B^{(k, \text{new})} \defeq B^{(k,1)} \cdot XY$. Any Pauli term $P^{(k-1,\ell)}$ involving $B^{(k-1,1)}$ after commuting by $g_k$ now a linear combination of said term and $B^{(k, \lambda_{\text{new}})}$ according to \cref{eq:commutation-relations}. A similar argument holds when the one basis $B^{(k-1,1)}$ term acts as $Y$. In the case that two basis terms act non-trivially on the first qubit, we can also enforce that one of the basis terms acts as $Z$ and the other acts as either $X$ or $Y$. Then, a similar argument enforces that it suffices to introduce a single additional basis term. 

Therefore, at the end of the induction, the total basis has size at most $t + 1$ and $U Z_1 U^\dagger$ can be expressed as a linear combination of $\leq 2^t$ Pauli terms.
\end{proof_of}

\end{document}